\documentclass[11pt]{article}
\usepackage[english]{babel}

\newcommand{\version}{}

\usepackage{amsthm,amsfonts,amsmath}
\usepackage{graphics
,
}
\swapnumbers
\theoremstyle{plain}

\newtheorem{thm}{THEOREM}[section]
\newtheorem{lm}[thm]{LEMMA}

\newtheorem{remark}[thm]{Remark}

\theoremstyle{definition}
\newtheorem{defi}[thm]{DEFINITION}
\theoremstyle{remark}
\newcommand{\upchi}{\raise1pt\hbox{$\chi$}}
\newcommand{\R}{{\mathord{\mathbb R}}}

\newcommand{\Z}{{\mathord{\mathbb Z}}}

\newcommand{\dd}{{\rm d}}

\renewcommand{\|}{{\Vert}}
\numberwithin{equation}{section}
\let \e=\varepsilon                                                                

\begin{document}
\markboth{\scriptsize{CEM \version}}{\scriptsize{CEM \version}}

\title{Analysis of Droplets in \\ Lattice Systems with Long-range Kac Potentials}
\author{\vspace{5pt} E. A.  Carlen$^1$, R. Esposito$^2$ J.L. Lebowitz$^3$ and R. Marra$^4$\\}

\date{\version}
\maketitle

\def\O{\Omega}

\footnotetext                                                                         
[1]{ Department of Mathematics, Rutgers University, 
Piscataway, NJ
08854, U.S.A. Work partially
supported by U.S. National Science Foundation
grant DMS 06-00037.  }
\footnotetext 
[2]{Dipartimento di
Matematica, Universit\`a dell'Aquila, Coppito, 67100 AQ, Italy}
\footnotetext                                                                         
[3]{ Departments of Mathematics and Physics, Rutgers University, 
Piscataway, NJ
08854, U.S.A. Work partially
supported by U.S. National Science Foundation
grant DMR  08-02120 and AFOSR Grant FA 9550-10-1-0131 }

\footnotetext 
[4]{Dipartimento di Fisica and Unit\`a INFN, Universit\`a di
Roma Tor Vergata, 00133 Roma, Italy. \\
\copyright\, 2005 by the authors. This paper may be reproduced, in its
entirety, for non-commercial purposes.}                  

\maketitle

\begin{abstract}
We investigate the geometry of typical equilibrium configurations for a lattice gas in a finite macroscopic domain with attractive, long range Kac potentials. We focus on the case when the system is below the critical temperature  and has  a fixed number of occupied sites.We connect the properties of  typical configurations to the analysis of the constrained minimizers of a mesoscopic non-local free energy functional, which we prove to be  the large deviation functional for a density profile in the canonical Gibbs measure with prescribed global density.
In the case in which the global density  of occupied sites lies between the two equilibrium densities
that one would have without a constraint on the particle number, a ``droplet'' of the high (low) density phase may or may not form in a background of the low (high) density phase. We determine the critical density for droplet formation, and the nature of the droplet, as a function of the temperature and the size of the system, by  combining the present large deviation principle with the analysis of the mesoscopic functional given in \cite{CCELM0}.%
\end{abstract}

\section{Introduction}\label{intro0}
The mathematical study of the behavior of a lattice gas of particles (or spins) interacting via a long range Kac potential, both in equilibrium and in non-equilibrium, has been the subject of many works:  a recent book  \cite{P}  provides   a comprehensive treatment of the subject as it has developed so far.

The  long-range Kac potential  introduces a third length scale between the microscopic scale  of the lattice spacing and the macroscopic scale  of the size of
the domain.  This third scale is  referred to as the mesoscopic scale.  As we show here, one can determine the geometric nature of ``typical''
microscopic particle configurations for such systems through the analysis of a mesoscopic free energy functional that serves as a large deviations functional for the system. We do this in a scaling regime that is critical for droplet formation in these models.  

The large deviations functional (LDF) with which we work  is a functional of ``course grained density profiles'',    different from the
LDF used to study the corresponding problem in two dimensional nearest-neighbor Ising systems:  There, the LDF is a function of ``contours''  associated to the microscopic configuration on a two dimensional lattice.  The mesoscopic analysis that we carry out here goes through in any dimension. 
Before describing our results, we first describe the models with which we work more precisely.

\subsection{The grand canonical and canonical measures for the model}

 Let ${\cal T}_L$ be the $d$--dimensional square torus with side length $L$. 
Let 
$\Lambda_{L,\gamma}$  denote the part of the lattice $\gamma \Z^d$ contained in ${\cal T}_L$.  For the purpose of connecting the microscopic and mesoscopic scales, it will be convenient, as in \cite{ABCP,P}, to regard
$\Lambda_{L,\gamma}$ as a subset of  ${\cal T}_L$. 
A {\it particle  configuration},  is a function $s$ 
from $\Lambda_{L,\gamma}$ to $\{-1,1\}$. The  site $x(i)=\gamma \, i\in \Lambda_{L,\gamma} \subset \mathcal{T}_L$,  $i\in \mathbb{Z}^d$  is {\it occupied} by a particle 
if $s(x(i)) =1$, and is {\it unoccupied}
if $s(x(i)) = -1$. 

The Hamiltonian $H_{\gamma,L}$ for the system, giving the total interaction energy of a configuration, is
\begin{equation}\label{hamdef}
 H_{\gamma,L}(s) = 
-\frac{\gamma^d}{2} \sum_{x(i),x(j)\in  \Lambda_{L,\gamma}}J(|x(i)-x(j)|)s(x(i))s(x(j))\ .
\end{equation}
where $J$ is a
a non negative smooth function $J$ on $\R_+$ that is bounded, continuous,  supported by $[0,1]$, and strictly monotone decreasing on $[0,1]$,
where $|x(i)-x(j)|$ denotes the distance separating $x(i)$ and $x(j)$ in the torus. 
We assume the normalization
\begin{equation}\label{aldef}
\int_{\R^d}J(|r|){\rm d}r = 1\ .
\end{equation}  
This function $J$ is the {\it interaction potential}. Since its range is of order $\gamma^{-1}$ in microscopic units, $ H_{\gamma,L}$ is a {\em local mean field} Hamiltonian:
A particle at site $x(i)=\gamma \/\/i \in \Lambda_{L,\gamma}$ interacts with a local mean field of neighboring particles:
 \begin{equation}\label{intdef}
 \gamma^{d}\sum_{\gamma \/\/j\in \Lambda_{L,\gamma}}J(|x(i)-x(j)|)s(x(j))\ .
 \end{equation}

Let $\Omega_{L,\gamma}=\{-1,1\}^{\Lambda_{L,\gamma}}$ denote the set of all particle configurations. Also, for any $0 \leq N \leq {|\Lambda|}$, where
$|\Lambda|$ denotes the number of sites in $\Lambda$, define 
$$\Omega_{L,\gamma,N}:= \{\  s\in \Omega_{L,j}\ :\  \sum_{x\in \Lambda_{L,j}}(s(x)+1)/2 = N\ \}\ .$$
Then $\Omega_{L,\gamma,N}$ is the space of $N$-particle configurations.  

Given an inverse temperature $\beta$,  the grand canonical Gibbs measure $P_{\rm gc}$ on $\Omega_{L,\gamma}$ is defined by
\begin{equation}
P_{\rm gc}(\{s\})=\frac{1}{Z_{\rm gc}} \exp[-\beta H(s)] \qquad{\rm and}\qquad Z_{\rm gc} = \sum_{s\in \Omega} \exp[-\beta H(s)]\ ,
\end{equation}\label{grcan}
and the canonical Gibbs measure $P_{\rm can}$ on $\Omega_{L,\gamma,N}$ is defined by
\begin{equation}\label{can1}
P_{{\rm can},N}(\{s\})=\frac{1}{ Z_{{\rm can},N} } \exp[-\beta H(s)] \qquad{\rm and}\qquad Z_{{\rm can},N} = \sum_{s\in \Omega_N} \exp[-\beta H(s)]\ ,
\end{equation}
where we have dropped the subscripts from $H_{\gamma,L}$. 

As is well known, our system undergoes a phase transition when $\gamma$ goes to zero  at $\beta = 1$. 
In particular, the nature of the microscopic configurations that are typical under the Gibbs measure
changes at the phase transition. To see this change clearly, 
it is convenient to introduce the notion of  {\it corse grained} configurations, $\sigma_\delta(r)$. We shall give a precise definition in the next subsection. For the moment it suffices to say that $\sigma_\delta(r)$ is obtained by averaging the microscopic configuration $s$ on a box of size $\gamma^\delta$ centered at $r$, $\delta<1$ so $\gamma^\delta\gg 1$ for $\gamma\ll 1$.  

For $\beta>1$, $\gamma$ small and $L$ large, and say $\delta = 1/2$,  the grand canonical probability measure is overwhelmingly concentrated on coarse grained configurations $\sigma_\delta$ for which either 
$\sigma_\delta(r)$ is very close to $m_\beta$ at most $r\in {\mathcal T}_L$ or else  $\sigma_\delta(r)$ is very close to $-m_\beta$ at most $r\in {\mathcal T}_L$, where
$m_\beta$ is the unique positive solution to 
\begin{equation}\label{mbdif}
m_\beta = {\rm tanh}(\beta m_\beta)\ .
\end{equation}
If $\sigma_\delta(r) \approx m_\beta$ we say that the system is in the high density, or ``liquid'' phase at $r$, and if 
$\sigma_\delta(r) \approx -m_\beta$ we say that the system is in the low density, or ``vapor'' phase at $r$.

 Things are more interesting under the canonical measure,  under  which the global  particle density has the sharp value $n= N/|\Lambda|$. If
 $$-m_\beta  < 2n-1 < m_\beta\ ,$$
 then it is not possible for the system to be in one phase or the other over all of ${\mathcal T}_L$. Instead, as one might expect,  typical configurations  will be such that 
 $\sigma_\delta(r) \approx +m_\beta$ in some part $D$ of ${\mathcal T}_L$, while
 $\sigma_\delta(r) \approx -m_\beta$ in most of the rest of ${\mathcal T}_L$.   If $n$ is much closer to $\frac{1-m_\beta}2$ than it is to $\frac{1+m_\beta}{2}$, we would expect the vapor state to dominate, so that
 $D$ will cover only a small part of the whole domain.  In such a configuration, we say there is a {\em droplet} of the liquid phase  in a background of the vapor phase. 
 The basic question that concerns us here is this:
 
 \medskip
 \noindent{$\bullet$} {\em For a given $n$, $\beta>1$, small $\gamma$  and large $\Lambda$, what are the sizes and shapes of droplets for typical configurations under $P_{{\rm can},N}$?}
 \medskip

To answer this question we first have to define precisely the coarse graining we will use.

\subsection{The coarse--graining transformation} \label{intro1c}

It will be convenient to assume, as in \cite{ABCP,P} that $\gamma = 2^{-k}$ for some positive integer $k$, and that $L$ is an integer. 
As in  \cite{ABCP,P}, we
regard particle configurations as functions on ${\cal T}_L$, and not only on the lattice 
$ \Lambda_{L,\gamma}$. However, to keep clear which variables range over ${\cal T}_L$, and which range over  $ \Lambda_{L,\gamma}$, we
use their convention of writing
$r$ to denote a continuous variable in ${\cal T}_L$, and $x$ to denote the discrete variables
on $\Lambda_{L,\gamma}$. 

The lattice induces a partition  $\mathcal{Q}^{(k)}$ of  ${\cal T}_L$ into cubes $Q(x)$, $x\in \Lambda_{L,\gamma}$: 
For any $x\in  \Lambda_{L,\gamma}$,
$$Q(x) = \{ r\in {\cal T}_L\ :\  x_j < r_j \leq x_j + \gamma\ , j=1,\dots\ d\}\ $$
We then extend the domain of each $s\in \Omega_{L,\gamma}$ from  $\Lambda_{L,\gamma}$ to  ${\cal T}_L$
by putting
$$s(r) = s(x) \qquad{\rm for\ all} \quad r \in Q(x)\ .$$

We are now ready to introduce the coarse-graining transformation:  Fix any positive integer $\ell$ with $\ell < k$.   Then the cubic partition
$\mathcal{Q}^{(\ell)}$ is coarser than  $\mathcal{Q}^{(k)}$:  Each cube in $\mathcal{Q}^{(\ell)}$ is the union of $2^{k-\ell}$ cubes in 
$\mathcal{Q}^{(k)}$.    The cubes in the partition  $\mathcal{Q}^{(\ell)}$ have side-length $\gamma^{\ell/k}$, and it is therefore
useful to introduce the parameter
$$\delta = \frac{\ell}{k}\ .$$

\begin{defi}[Coarse graining transformation on scale $\delta$]  Let $f$ be any integrable function on  ${\cal T}_L$, and let $\delta = \ell/k$
for some integers $0 < \ell < k$. The coarse grained projection of $f$ on scale $\delta$ is the function on  ${\cal T}_L$ given by
\begin{equation}\label{coarse}
\pi^{(\delta)}f(x)=\frac 1{|C(x)|}\int_{C(x)}f(y)d\/y,
\end{equation}
where $C(x)$ is the unique cube in  $\mathcal{Q}^{(\ell)}$ that  contains $x$.  In other words \cite{P},   $\pi^{(\delta)}f$ is the conditional expectation of 
$f$ given the $\sigma$-algebra generated by  $\mathcal{Q}^{(\ell)}$. 
\end{defi}

Since we consider particle configurations in $\Omega_{L,\gamma}$ as functions, necessarily integrable,  on  ${\cal T}_L$, the coarse graining
transformation may be applied to each $s\in  \Omega_{L,\gamma}$.

\begin{defi}[The coarse grained configuration spaces] For any $0 < \ell < k$, and with $\delta = \ell/k$, 
let $\Omega_{L,\gamma}^{(\delta)}$ be the set of  functions $\sigma_{\delta}=\pi^{(\delta)}\sigma$ for some $\sigma\in \Omega_{L,\gamma}$. 
That is, $\Omega_{L,\gamma}^{(\delta)}$   is the image of $\Omega_{L,\gamma}$ under $\pi^{(\delta)}$. Likewise, define
 $\Omega_{L,\gamma,N}^{(\delta)}$   is the image of $\Omega_{L,\gamma,N}$ under $\pi^{(\delta)}$
\end{defi}

Note that the elements of   $\Omega_{L,\gamma,N}^{(\delta)}$ are not only constant on each cube $\mathcal{Q}^{(\ell)}$, but they can only assume a finite, discrete set of values: for all $r$, 
$$\sigma_{\delta}(r) \in \{m \gamma^\delta\ ,\ m =0,1\dots,  \gamma^{-\delta}\ \}\ .$$

For each given $\sigma_{\delta}\in \Omega_{L,\gamma}^{(\delta)}$ we consider the event
$$E(\sigma_{\delta})=\{s \in \Omega_{L,\gamma} \,|\, \pi^{(\delta)}s=\sigma_{\delta}\}.$$
Following \cite{P}, we define
$$Z(\sigma_\delta) = \sum_{s\in  E(\sigma_\delta)} e^{-\beta H(s)}\ .$$

The grand canonical probability of  $E(\sigma_{\delta})$ is given by
\begin{equation}\label{ratio1}
P_{\rm gc}[E(\sigma_\delta)] = \frac{Z(\sigma_\delta)}{Z_{{\rm gc}}}\ .
\end{equation}
Furthermore, provided
${\displaystyle \gamma^{-d}\int_{{\mathcal T}_L}\sigma_\delta(r)\dd r  = N}$ so that $ E(\sigma_\delta)\subset \Omega_{L,\gamma,N}$,
and hence so that $P_{{\rm can},N}[E(\sigma_\delta)]$ is defined,
\begin{equation}\label{ratio2}
P_{{\rm can},N}[E(\sigma_\delta)] = \frac{Z(\sigma_\delta)}{Z_{{\rm can},N}}\ .
\end{equation}

\subsection{The free energy functional}

The key to the solution of the droplet problem is to relate the probabilities $P_{\rm gc}[E(\sigma_\delta)]$ and $P_{\rm can,N}[E(\sigma_\delta)]$
to a mesoscopic free energy functional of $\sigma_\delta$. 

\begin{defi}[The GPL free energy functional] Let $\mathcal{M}$ be the set of  measurable function $\sigma$ mapping $\mathcal{T}_L$ into $[-1,1]$.
 For any $\sigma\in \mathcal{M}$, The Gates-Penrose-Lebowitz free energy of $\sigma$,   ${\mathcal F}(\sigma)$  is defined   
by
\begin{eqnarray}\label{free-energy-funct}
\mathcal{F}(\sigma)
&=&\frac{1}{\beta}\int_{\mathcal{T}_L} i(\sigma(x))\dd x - \frac {1}{2}  \int_{\mathcal{T}_L}\int_{\mathcal{T}_L} J(|x-y|)\sigma(x)\sigma(y)\dd x\dd y\nonumber\\
&=&\int_{\mathcal{T}_L} f(\sigma(x))\dd x + \frac {1}{4}  \int_{\mathcal{T}_L}\int_{\mathcal{T}_L} J(|x-y|)[\sigma(x)-\sigma(y)]^2\dd x\dd y\ 
\end{eqnarray}
where
$$f(m) :=  \frac{1}{\beta}i(m) -\frac{1}{2}m^2\ ,$$
and
$$i(m) 
=\frac{1-m}{2}\log\frac{1-m}{2}+\frac{1+m}{2}\log\frac{1+m}{2}\ .$$
\end{defi}

We shall prove the following theorem:

\begin{thm}\label{propos1}  For any $-1 \leq m \leq 1$, let $N = [\frac{m+1}2 \gamma^{-d}L^d]$, recalling that
$|\Lambda_{L,\gamma}| = \gamma^{-d}L^d$ is the total number of lattice sites.

Let $\sigma^\star$ be a  constrained minimizer of $\mathcal{F}$:
\begin{equation}\label{conmin}\frac 1{L^d}\int_{\mathcal{T}_L} \sigma^\star(r)\dd r=m  \quad{\rm and}\quad 
{\mathcal F}(\sigma^\star) =  \inf\left\{ {\mathcal F}(\sigma)\ :\  \frac 1{L^d}\int_{\mathcal{T}_L} \sigma(r)\dd r=m \right\}\ . 
\end{equation}
Then there is $\delta>0$ such that
\begin{equation}
| \log P_{{\rm can},N} [E(\sigma_{\delta})| +\beta\gamma^{-d}[\mathcal{F}(\sigma_{\delta})-\mathcal{F}(\sigma^\star)]|\le c\gamma^{-d}L^{d}  [\gamma^\delta +\gamma^{d(1-\delta)}]\log\gamma^{-1},
\end{equation}
Moreover,  let $\mathcal{A}$  be any set of coarse-grained configurations in $\Omega_{L,\gamma,N}^\delta$. Then 
$$\gamma^{d} \log P_{{\rm can},N}(\mathcal{A})= -\inf_{\sigma\in \mathcal{A}}\beta\Big[\mathcal{F}(\sigma)-\mathcal{F}(\sigma^\star)\Big] + L^{d}\mathcal{O}(\gamma^{\bar\delta}).
$$
\end{thm}

Analogous results are proved in \cite{ABCP,P} for the grand canonical measure $P_{\rm gc}$. Note that when $E(\sigma)\subset \Omega_{L,\gamma,N}$,
the difference between  $P_{{\rm can},N}(E(\sigma))$ and  $P_{{\rm gc}}(E(\sigma))$ is entirely in the denominators. Thus, all of the estimates in
 \cite{ABCP,P} on the numerator in  (\ref{ratio1}) apply immediately to the numerator in (\ref{ratio2}).  

 However, there is  somewhat more to be done to estimate the denominator in  (\ref{ratio2}) than the denominator in  (\ref{ratio1}).  
 
 The reason is that in estimating
 the denominator in  (\ref{ratio1}), one shows that the main contribution comes from a small number of configurations that are uniformly close to one of the constant profiles
 $\sigma = \pm m_\beta$ that are the global minimizers of ${\mathcal F}$.   In the canonical case, we are concerned with profiles that are uniformly close to 
  some profile $\sigma^\star$ that minimizes the constrained variational problem in (\ref{conmin}). 
  
  We do not have {\it a-priori} information on what these are as we do in the grand canonical case. In particular, we do not know {\it a-priori} that they are bounded away from $\pm 1$, 
 where the derivative of the function $f$ emerging in the definition of ${\mathcal F}$ is infinite.  Thus more work has to be done to show that all such configurations have a free energy ${\mathcal F}(\sigma)$ that is very close to  ${\mathcal F}(\sigma^\star)$.
 This is the main difference between what we do here, and what was already done in the grand canonical case. The main tool that enables us to deal with this point is
 Theorem~\ref{nearlip}.
 
The issue of large deviations in the canonical ensemble has been recently addressed by Bertini et al. \cite{BCP} who have proved, using a different technique, a large deviation result for density profiles in the canonical ABC model; a one-dimensional three species mean field model which exhibits coexistence of phases at low temperature.

 \subsection{Results}
 
Our main result is  that if $\beta>1$, $\delta$ small,   $L$ large, and $(m+m_\beta) = {\mathcal O}\left(L^{-\frac{d}{d+1}}\right)$, 
then, with extremely high probability the coarse-grained profile has a droplet of $+$ phase of a particular size. 
To make this precise we begin with some definitions. 

Set $\kappa=(m+m_\beta)^{1/3}\asymp L^{-\frac{d}{3(d+1)}},$
\begin{equation}\label{hpmdef}
h_+ = m_\beta -\kappa\qquad{\rm and}\qquad h_- = -m_\beta + \kappa\ .
\end{equation}
Note that these two values are just below $m_\beta$, and just above $-m_\beta$.   If a coarse grained configuration $\sigma$ has $\sigma(r) \ge h_+$, then
the configuration is dominated by the ``liquid'' state at $r\in {\mathcal T}_ L$, while if
$\sigma(r) \le h_-$, then
the configuration is dominated by the ``vapor'' state at $r$. By definition, {\em  our droplet of the liquid state} for the coarse-grained configuration $\sigma$
 is the region  in which $\sigma\ge h_+$. For each coarse grained configuration $\sigma$, we define the sets:
\begin{eqnarray}\label{ABC}
 A(\sigma) &=& \{\ x\in {{\cal T}_L}\ :\ h_- \le \sigma(r) \le h_+\ \}\nonumber\\ 
 B(\sigma) &=&  \{\ x\in {{\cal T}_L}\ :\ \sigma(r) \le h_-\ \}\nonumber \\
C(\sigma) &=&  \{\ x\in {{\cal T}_L}\ :\  \sigma(r) \ge h_+\ \}. \label{ABCdef}
\end{eqnarray}

According to Lemma 4.6 of \cite{CCELM0}, provided that $\mathcal{F}(\sigma)\le \mathcal{F}(n)$,  the size of $A$ is given by 
$$|A| \le c  \kappa^4L^d\asymp L^{\frac{3d^2-d}{3(d+1)}}\ .$$
For  $L$ large this is a negligibly small fraction of $D_0:=\frac{m_\beta+m}{2m_\beta}L^d$, the {\it equimolar volume}. We define the {\it droplet volume} of a microscopic configuration $\sigma$ to be $|C(\sigma)|$ for the coarse grained profile $\pi^{(\delta)}(\sigma)$. Moreover, we introduce the volume fraction 
$$\eta(\sigma)=\frac{|C(\sigma)|}{D_0}.$$
A more precise version of the statement made in the beginning of this subsection will  be that with high probability $\eta$ is close to a critical value $\eta_c$, which is computed by minimizing a suitable function given below, see Theorems \ref{teorema} and \ref{teorema2}. 

The rest of the paper is organized as follows: In the next section, we use Theorem~\ref{propos1}
to prove a theorem that makes precise the statements made above about droplets, and we discuss what else needs to be done to prove a result specifying the shapes of
typical droplets.  Then in Section \ref{Sect3}, we give the proof of Theorem~\ref{propos1} and Theorem~\ref{nearlip}, on which it depends. Finally, in Section \ref{Sect4} we discuss some open problems.

\section{Typical microscopic configurations}\label{Sect2}

The minimizers for the functional (\ref{free-energy-funct}) for general smooth functions $\sigma$ satisfying  the constraint 
\begin{equation}\label{costr}
L^{-d}\int_{{\mathcal T}_L}d x \sigma(x)=m,
\end{equation}
have been  studied  in \cite{CCELM0} under the assumption 
\begin{equation}\label{criden} m=-m_\beta +KL^{-\frac{d}{d+1}},\end{equation}
for $K>0$. Here $\frac d{d+1}$ is the critical scaling for droplet formation. Indeed, it turns out that there is a critical value $K_*$ for $K$, such that for $K<K_*$ there is no droplet formation, while, for $K>K_*$ a droplet will form. 

Informally speaking, it is proved that 
\begin{eqnarray}&&
\inf \left\{ {\cal F}(m)\ :\ \frac{1}{L^d}\int_{{\cal T}_L}m(x){\rm d}x =m\ \right\} \approx \inf_{\eta\in [0,1]}\Phi(\eta),\nonumber\\
&&\Phi(\eta):=L^{\frac{d^2-d}{d+1}} Sd\omega_d\left(\frac{ Kd}{2m_\beta d\omega_d}\right)^{1-\frac{1}{d}}
\left[\eta^{1-1/d} + H(K)(1 - \eta)^2\right]\ ,
\end{eqnarray}
where 
\begin{equation}\label{dkd}
H(K) = 
\frac{2m_\beta^2}{d\chi S}\left(\frac{d}{d\omega_d}\right)^{\frac{1}{d}}
\left(\frac{K}{2 m_\beta}\right)^{1+\frac{1}{d}}\ .
\end{equation}

\medskip
Here $S$ is the planar surface tension, $\chi$ is the compressibility and $\omega_d$ the volume   of the ball in $\mathbb{R}^d$. 
(See \cite{CCELM0} for the definition of $S$ and $\chi$; for our purposes now, they are some computable constants associated to the model.)
We define $\eta_c\ge 0$ as the absolute minimizer of the function $\Phi(\eta)$.

The precise estimate we use here is (5.13) from  \cite{CCELM0} which says:

\begin{equation}\label{bound69}\mathcal{F}(\sigma)-\mathcal{F}(\sigma^\star)\ge [\Phi(\eta)-\Phi(\eta_c)](1+o(L^{-\frac{d^2-d}{d+1}})).
\end{equation}
\begin{thm}\label{teorema}
For any $L$ and any $K>0$ let $m$ be  given by $m=- m_\beta + K L^{-\frac{d}{d+1}}$. For any $\alpha>0$ and any microscopic configuration $\sigma$ such that
 \begin{equation}\label{const7}
 m = L^{-d}\int_{{\cal T}_L}\sigma(x){\rm d}x\ ,
\end{equation}
there is a universal constant $M$ such that, for $L$ large enough,
\begin{equation}\label{3.9}
\mathcal F(\sigma)\ge(1+\alpha) \mathcal F(\sigma^\star)
\end{equation}
whenever 
\begin{equation}\label{3.10}\frac{(\eta(\sigma)-\eta_c)^2}{M}\ge\alpha\ .\end{equation}

\end{thm}
\begin{proof}
We take the quadratic approximation of the function $\varphi(\eta)=\eta^{1-1/d} + H(K)(1 - \eta)^2$ around its minimizer $\eta_c$. The inequality (\ref{3.9}) then follows from (\ref{bound69}) and the definition of $\Phi$ since $\mathcal{F}(\sigma^\star)=\Phi(\eta_c)$.
\end{proof}

\begin{thm}\label{teorema2}
Let $n$ and $\sigma$ satisfy the hypothesis of theTheorem~\ref{teorema}.  Define the event
\[\mathcal{A}_{\e}=\{\sigma\in \Omega^\ell_{L,\gamma}\ : |\eta(\sigma)-\eta_c|>\e\}.\]
Then, for $L$ large enough and $\gamma^{\bar\delta}\le \frac {\e^2} {6C} \beta \Phi(\eta_c)$, with $\bar\delta$ in Theorem~\ref{propos1} and $M$ as in  Theorem \ref{teorema},
\begin{equation}
\gamma^d\log P(\mathcal{A}_{\e})\le -\frac{\e^2}{2M}\beta \Phi(\eta_c).
\end{equation}
\end{thm}
\begin{proof}
By the definition of $\mathcal{A}_\e$ and Theorem \ref{teorema},
\[\inf_{\sigma\in \mathcal{A}}\beta\Big[\mathcal{F}(\sigma)-\mathcal{F}(\sigma^\star)\Big]\ge \frac{\e^2}{M}\beta\mathcal{F}(\sigma^\star).\]
By Theorem 2.1 in \cite{CCELM0}, 
\[\mathcal{F}(\sigma^\star)=\Phi(\eta_c)(1 +o(L^{-\frac{d^2-d}{d+1}}))\]\ .
Combining this  with Theorem \ref{propos1} we obtain 
\[
\gamma^{d} \log P(\mathcal{A}_\e)\le - \frac{\e^2}{M}\beta\Big[
\Phi(\eta_c) +o(L^{-\frac{d^2-d}{d+1}})
\Big]+L^{d}\mathcal{O}(\gamma^{\bar\delta}).
\]
For $L$ large enough we
 have \[
 \Phi(\eta_c)(1 +o(L^{-\frac{d^2-d}{d+1}}))\ge \frac 2 3 \Phi(\eta_c).\]
Hence, by choosing $\gamma^{\bar\delta}\le \frac {\e^2} {6M} \beta \Phi(\eta_c)\asymp \e^2L^{-\frac{2d}{d+1}}$ we conclude the proof.
\end{proof}

\begin{remark}{\rm 
If $K<K_*$, then $\eta_c=0$ and there are no droplets. If $K>K_*$, then $\eta_c>0$ so there is droplet formation. The case $K=K_*$ cannot be decided because there are two absolute minimizers, one corresponding to $\eta_c=0$ and another to $\eta_c>0$, which give the same value to $\Phi(\eta)$.}
\end{remark}

\section{Canonical Large Deviations}\label{Sect3}

This section provides the proof of the canonical ensemble large deviations bounds we use. We closely follow \cite{ABCP,P} where possible. As noted above, the main difference is in the estimation of the denominator in (\ref{ratio2}).

We now prove a continuity estimate for the functional $\mathcal{F}$ that will show that for any profile $\sigma$ that is close to a minimizer in the $L^\infty$ norm,
$\mathcal{F}(\sigma)$ has very nearly the minimal value.  This is the key to the estimation of the denominator in (\ref{ratio2}).

Though we have a good knowledge of the minimizers for certain values of the constraint,  we do not
have this knowledge for other values.  In general, there is no rigorous {\it a-priori} argument to exclude the possibility that a minimizing profile takes values very
close to $\pm 1$ on sets of significant size.  This would cause  a difficulty since $|f'(m)|$, the absolute value of the derivative of $f$, tends to infinity as $m$ tends to $\pm 1$. 
Nonetheless, the functional  $\mathcal{F}(\sigma)$ is {\it nearly} Lipschitz continuous on its entire domain.

\begin{thm}[Near Lipschitz continuity of $\mathcal{F}$]\label{nearlip}
Let $\sigma$ and $\sigma_0$ be two functions on $\mathcal{T}_L$ with values in $[-1,+1]$, and suppose that
$$\|\sigma - \sigma_0\|_\infty \leq h\ $$
with $2h \leq 1 - m_\beta$. 
Then there is a universal constant $C$ so that 
$$|\mathcal{F}(\sigma) - \mathcal{F}(\sigma_0)| \leq CL^d h |\log(h)|\ .$$
\end{thm}

\medskip

\begin{proof}  Let $\varphi(x) :=  \sigma(x) - \sigma_0(x)$, $$A_+ = \{\ x\ : \sigma_0(x) \geq 1 - 2h\}\qquad{\rm and}
\qquad A_- = \{\ x\ : \sigma_0(x) \leq -1 + 2h\} \ .$$ Finally, let $B$ denote the complement of $(A_+\cup A_-)$.
We now seek an upper bound on $\mathcal{F}(\sigma) - \mathcal{F}(\sigma_0)$.

Since $i$ is monotone increasing in the set $[m_\beta,1]$, and due to our restriction on $h$, 
on the set $A_+$,
$$i(\sigma(x)) - i(\sigma_0(x)) = \leq i(1) - i(1 - 2h)\ .$$
Likewise, on 
the set $A_-$,
$$i(\sigma(x)) - i(\sigma_0(x))  \leq i(-1) - i(-1+ 2h) = i(1) - i(1 - 2h)\ .$$
Note that
$$ i(1) - i(1 - 2h) \leq Ch|\log (h)|\ .$$
Next, on the set $B$, 
\begin{eqnarray}
i(\sigma(x)) - i(\sigma_0(x)) &=&  i(\sigma_0(x)+ \varphi(x)) - i(\sigma_0(x)) \nonumber\\
&\leq & \left(\sup_{-1 + h \leq m \leq 1-h}|i'(m)|\right)h\ .\nonumber\\
\end{eqnarray}
Note that
$$\sup_{-1 + h \leq m \leq 1-h}|i'(m)| =   |i'(1-h)| \leq C|\log(h)|\ .$$
Thus,
$$\int_{\mathcal{T}_L} [i(\sigma(x)) - i(\sigma_0(x))]{\rm d}x \leq  C L^d h |\log(h)|\ .$$
The interaction term poses no problem:
\begin{multline}\sigma(x)J(|x-y|)\sigma(y) -  \sigma_0(x)J(|x-y|)\sigma_0(y) \\=  \varphi(x)J(|x-y|)\sigma(y) + \sigma_0(x)J(|x-y|)\varphi(y)\ ,\end{multline}
and so the interaction term is Lipschitz. Thus, we have the upper bound
$$\mathcal{F}(\sigma) -  \mathcal{F}(\sigma_0) \leq CL^d h |\log(h)|\ .$$
By the symmetry of the hypotheses, we have the same upper bound for 
$\mathcal{F}(\sigma_0) -  \mathcal{F}(\sigma)$, and this proves the Theorem.
\end{proof}


\subsection{Proof of Theorem~\ref{propos1} } \label{intro1c}
\medskip

For each given $\sigma_{\delta}\in \Omega_{L,\gamma}^{(\delta)}$ we consider the event
$E(\sigma_{\delta})=\{s\in \Omega_{L,\gamma} \,|\, \pi^{(\delta)}s=\sigma_{\delta}\}$.
Given $N$ and thus $n= N/|\Lambda_{L,\gamma}|$.
We denote by $P(\sigma_{\delta})$ the canonical probability
$P_{{\rm can},N}[E(\sigma_{\delta})]$.  We quote the following lemma from \cite{ABCP}:

\begin{lm}\label{lemma1} Let $W_\ell(\sigma_{\delta})$ be the cardinality of $E(\sigma_{\delta})$. Then\begin{equation}
\left|\log W(\sigma_{\delta})+ \gamma^{-d}\int_{\mathcal{T}_L} i(\sigma_{\delta}(x))\right|\le c L^d\ell d2^{d(k-2\ell)}.
\end{equation}
\end{lm}

\begin{proof} Let $C$ be any atom in the coarse-graining partition $\mathcal{Q}^{(\ell)}$ and $\bar \sigma$ be the average of $\sigma$ on $C$. There are $N_C=2^{dl}$ atoms of the fine partition $\mathcal{Q}^{(k)}$ in $C$ and $\sigma(x)=-1$ on exactly $K(\bar \sigma)=\frac{1-\bar \sigma}2 N$ 
 cubes of $\mathcal{Q}^{(k)}$ in $C$. Therefore, the number of compatible microscopic configurations in $C$ is given by $N_C\choose K(\bar \sigma)$ and, by straightforward Stirling analysis, one can check that there is a constant $c$ such that
 $$\left|\log{{N_C}\choose{K(\bar \sigma)}}+ i(\bar \sigma)\right|\le c'\frac{\log N_C}{N_C}= {c}\frac {d\ell}{2^{d\ell}}.$$
Summing over the $C$'s in $\mathcal{Q}^{(\ell)}$  gives the result.
\end{proof}

\begin{lm}\label{lemma2}
Let $N_\ell$ be the cardinality of $\Omega_{L,\gamma}^{(\delta)}$. Then there is a constant $c$ such that
$$\log N_\ell\le c L^d\ell d2^{d(k-\ell)}$$
\end{lm}

\begin{proof} The number of distinct values possible for $\bar \sigma$ is $2^{d\ell}$ and there are $L^d 2^{d(k-\ell)}$ coarse-grained cells in the torus. \end{proof}

\begin{lm}\label{lemma3}
There is a constant $c$ such that for each $\sigma\in \Omega_{L,\gamma}$ 
$$\Big|H(\sigma)-H(\pi^{(\delta)}\sigma)\Big|\le cL^d2^{\ell-k}$$
\end{lm}

\begin{proof}
By the definitions of $H$ and $\pi^{(\delta)}$ we get:
$$H(\sigma)-H(\pi^{(\delta)}\sigma)=\int_{\mathcal{T}_L} d\/x\int_{\mathcal{T}_L} d\/y \sigma(x)\sigma(y)\Big[J(|x-y|)- \pi^{(\delta)}\otimes\pi^{(\delta)}J(|x-y|)\Big].$$
But
$$\Big|J(|x-y|)- \pi^{(\delta)}\otimes\pi^{(\delta)}J(|x-y|)\Big|\le \|\nabla J\|_\infty \sqrt d 2^{\ell-k},$$
since $\sqrt d 2^{\ell-k}$ is the diameter of the cube $C\in \mathcal{Q}^{(\ell)}$.
\end{proof}

\begin{lm}\label{lemma4}
Let $\phi(r) $ be any continuously differentiable  function on $(\mathcal{T}_L)$, taking values in $[-1,1]$. Then there exists $\sigma_*\in \Omega_{L,\gamma}^{(\delta)}$ such that
$$\|\sigma_*-\phi\|_\infty\le c\|\nabla \phi\|_\infty [\gamma^\delta +\gamma^{d(1-\delta)}]$$
and
$$\int_{\mathcal{T}_L} \sigma_*(r)\dd r=\int_{\mathcal{T}_L} \phi(r)\dd r$$
\end{lm}

\begin{proof}
The construction is based on three steps:
\begin{enumerate}
\item Replace $\phi$ by $\pi^{(\delta)}\phi$;
\item For each $C\in \mathcal{Q}^{(\ell)}$, replace the value of $\pi^{(\delta)}\phi$ in $C$ by the closest value in the set of admissible values for the coarse-grained configurations, which are $-1+\frac{2j}{2^{d\ell}}$, with $j$ an integer between $0$ and $2^{d\ell}$;
\item If the coarse-grained configuration produced in step 2 has too high an average to satisfy the constraint, we lower the values of $\sigma_{\delta}$ on the necessary fraction of $C$ by an amount $2^{-d\ell-1}$ or rise it if it is too low.
\end{enumerate}
The first step  does not change the average and we have,
$$\|\pi^{(\delta)}\phi-\phi\|_\infty\le \|\nabla \phi\|_\infty 2^{\ell-k}\sqrt{d}.$$

In the second and third steps we shift the value by at most $2^{-d\ell-1}$. Putting these things together we get the proof of the lemma.
\end{proof}
\bigskip

\bigskip

\noindent {\bf Proof of Proposition \ref{propos1}.}
Let us pick a coarse-grained configuration $\sigma_{\delta}\in \Omega_{L,\gamma}^{(\delta)}$ such that
$$\frac 1 {L^d}\int_{\mathcal{T}_L} d\/x \sigma_{\delta}(x)= \frac{[mN]}N.$$
Then 
$$Z_{\beta,\gamma,L,m}P(\sigma_{\delta})=\sum_{\sigma\in E(\sigma_{\delta})}\text{\rm e}^{-\beta \gamma^{-d}H(\sigma)}
$$
We will make use of the parameter 
\begin{equation}\label{delta}
\delta= \frac{k-\ell}{k}
\end{equation}
and define
\begin{equation}\label{psi}
\psi(\delta)= cL^d\gamma^{-d}\Big\{\beta\gamma^\delta +\gamma^{2d(1-\delta)}\log(\gamma^{-1})\Big \}
\end{equation}
By Lemma \ref{lemma1} and Lemma \ref{lemma3}, we get:
$$\text{\rm e}^{-\beta \gamma^{-d}\mathcal{F}(\sigma_{\delta})-\psi(\delta)}\le\sum_{\sigma\in E(\sigma_{\delta})}\text{\rm e}^{-\beta \gamma^{-d}H(\sigma)}\le \text{\rm e}^{-\beta \gamma^{-d}\mathcal{F}(\sigma_{\delta})+\psi(\delta)}$$

We have to consider the partition function. In fact, the essential part of it comes from the microscopic configurations whose coarse-grained configuration corresponds to the minimizer of the free  energy.

Let $\sigma^\star$ be the minimizer of the free energy under the mass constraint and note that it fulfills the conditions required for the application of Lemma \ref{lemma4}. Indeed one gets immediately that the minimizer is of bounded variation, so that $\|\nabla \sigma^\star\|_1$ is bounded. Moreover, by using the Euler Lagrange equation and a bootstrap argument, it can be easily proved that $\|\nabla \sigma^\star\|_\infty$ is bounded.

Let $\sigma_{\delta}^*$ be the corresponding coarse-grained configuration provided by Lem\-ma \ref{lemma4}. Then
$$Z_{\beta,\gamma,L,m}\ge \text{\rm e}^{-\beta \gamma^{-d}\mathcal{F}(\sigma_{\delta}^*)-\psi(\delta)}.$$
Similarly, we have 
$$Z_{\beta,\gamma,L,m}\le \text{\rm e}^{-\beta \gamma^{-d}\mathcal{F}(\sigma_{\delta}^*)+\psi(\delta)}N_\ell.$$
By Lemma \ref{lemma2}
$$Z_{\beta,\gamma,L,m}\le \text{\rm e}^{-\beta \gamma^{-d}\mathcal{F}(\sigma_{\delta}^*)+\phi(\delta)}$$
with
$$\phi(\delta)=cL^d\gamma^{-d}[\gamma^d+\gamma^{d(1-\delta)}\log\gamma^{-1}].$$
In conclusion,
$$-\gamma^{-d}\beta [\mathcal{F} (\sigma_{\delta})-\mathcal{F} (\sigma_{\delta}^*)] -\phi(\delta)\le\log P(\sigma_{\delta})\le -\gamma^{-d}\beta [\mathcal{F} (\sigma_{\delta})-\mathcal{F} (\sigma_{\delta}^*)] +\psi(\delta)$$
We now replace $\mathcal{F}(\sigma_{\delta}^*)$ by $\mathcal{F}(\sigma^\star)$ in the estimates. For the upper bound, it is enough to use
$$ \mathcal{F}(\sigma_{\delta}^*)\ge \mathcal{F}(\sigma^\star).$$
For the lower bound instead we use Theorem \ref{nearlip}:
$$ |\mathcal{F}(\sigma_{\delta}^*)-\mathcal{F}(\sigma^\star)|\le c\gamma^{-d}L^{d}|\sigma_{\delta}^*-\sigma^\star|_\infty\log|\sigma_{\delta}^*-\sigma^\star|_\infty
$$
By Lemma \ref{lemma4} then we obtain 
$$ |\mathcal{F}(\sigma_{\delta}^*)-\mathcal{F}(\sigma^\star)|
\le c\gamma^{-d}L^{d} [\gamma^\delta +\gamma^{d(1-\delta)}]\log\gamma^{-1}$$
for a suitable constant $c$ depending on $\|\nabla \sigma^\star\|_\infty $.

The final estimate  is 
$$-\gamma^{-d}\beta [\mathcal{F} (\sigma_{\delta})-\mathcal{F} (\sigma^\star)] -c\gamma^{-d} L^{d} [\gamma^\delta +\gamma^{d(1-\delta)}\log\gamma^{-1}]
\le\log P(\sigma_{\delta})$$
$$\le -\gamma^{-d}\beta [\mathcal{F} (\sigma_{\delta})-\mathcal{F} (\sigma^\star)] +c\gamma^{-d}L^{d}    [\gamma^\delta +\gamma^{2d(1-\delta)}\log\gamma^{-1}]$$
\qed


\section{Concluding remarks}\label{Sect4}
\subsection{About the shape problem}

In the analysis in \cite{CCELM0} that leads to the crucial estimate (\ref{bound69}),  use was made of the Riesz Rearrangement Inequality for convolutions.
For any  measurable function $m$ on $\R^d$ such that the Lebesgue measure $\{x :\ m(x) > \lambda\}$ tends to $0$ and $\lambda$ tends to $\infty$,
let $m^*$ denote the radial function on $\R^d$ such that for all $\lambda>0$, the sets $\{x\ :\ m(x) > \lambda\}$ and 
$\{x\ :\ m^*(x) > \lambda\}$ have the same Lebesgue measure. (If the measure is infinite, the second set is all of $\R^d$.) Then the  Riesz Rearrangement Inequality says that
\begin{equation}\label{Reisz}  \int_{\R^2}\int_{\R^2} | m(x) -  m(y)|^2J(x-y)\dd x\dd y \geq
 \int_{\R^2}\int_{\R^2} | m^*(x) -  m^*(y)|^2J(x-y)\dd x\dd y\ .
 \end{equation}
In order to apply this to our profiles on the torus, we need to make some modifications of the profiles, ``lifting'' from ${\cal T}_L$ to $\R^d$, without significantly affecting the value of ${\mathcal F}$, as explained in
\cite{CCELM0}.  

Once this is done,   the rearrangement operation lowers the value of the free energy functional -- because of the Riesz Rearrangement Inequality -- and it makes the trial function $m$ radial. This facilitates the estimation of  ${\mathcal F}(m)$, and leads to (\ref{bound69}), but in the process, all information about the shape of the set $C(\sigma)$
defined in (\ref{ABC}), i.e., the droplet, is lost. 

To solve the shape problem, we would like to know, quantitatively, {\em how much} the rearrangement operation lowers the free energy for profiles $m$
in which the droplet is not nearly spherical. In purely mathematical terms, the question to be answered is this: Let $A$ be a measurable set in $\R^d$
with finite Lebesgue measure. Let $B$ be the ball with the same Lebesgue measure as $A$.  The {\em Fraenkel asymmetry} $F(A)$ is defined by
$$F(A) = \inf_{y\in \R^d}\left\{ \int_{\R^d} \left|1_A(x) - 1_B(x-y)\right|{\rm d}x \ \right\}\ .$$
It measures ``how out of round the shape of $A$ is''.

One would then like to have an explicit lower bound on 
\begin{equation}\label{Reisz1}  \int_{\R^2}\int_{\R^2} | 1_A(x) -  1_A(y)|^2J(x-y)\dd x\dd y -
 \int_{\R^2}\int_{\R^2} | 1_B(x) -  1_B(y)|^2J(x-y)\dd x\dd y\ 
 \end{equation}
 in terms of $F(A)$ and $J$.  This may be seen as a ``non-local isoperimetric inequality''.  Indeed, if $J$ is supported in a ball of radius $r>0$,
 then $| 1_A(x) -  1_A(y)| =0$ unless both $x$ and $y$ are within a distance $r$ of the boundary of $A$. 

\subsection{Local interactions}

We mention here that, as noted in \cite{CCELM0}, we expect the results derived there and here to apply to more general forms of $f(m)$ than that given in (\ref{free-energy-funct}). More precisely, the $\beta^{-1}i(m)$ in (\ref{free-energy-funct}) is the free energy of a lattice gas without any short range interaction. However the result should also be valid when, in addition to the Kac potential, also there are short range interactions, i.e. ones which do not scale with $\gamma$, as long as we are at values of $\beta$ where these do not, by themselves, produce a phase transition. This corresponds to replace, in the free energy functional, $\beta^{-1}i(m)$ by a strictly convex function $f_0(\beta,m)$, the free energy of the ``reference system'', as in references \cite{GP} and \cite{LP}. The technical problem in considering such system are the estimates as in Lemma \ref{lemma1}. To obtain such estimates for systems with short range interactions requires an estimate of the finite size corrections to $f_0(\beta,m)$. We know that in the canonical ensemble they go to zero as $\gamma\to 0$. One expects that they behave like the ratio of surface area to volume, i.e. of $O(\gamma L^{-1})$. It is an open problem to prove that they go as a suitably small power of $\gamma L^{-1}$. More is known in the grand-canonical ensemble for $\beta$ small and we hope that the same is true for the canonical case, whenever $f_0(\beta,m)$.

\medskip\medskip\medskip
\centerline{\bf Acknowledgments}
\medskip

\noindent We thank Errico Presutti for very helpful discussions. E. C., R. E. and R. M. acknowledge the kind hospitality of the Institute for Mathematical Science on the National Singapore University, where part of this work was done.


\medskip


\medskip

\begin{thebibliography}{99}

\bibitem{ABCP} Alberti G., Bellettini G.,  Cassandro M. and
Presutti  E., {\it Surface
Tension in Ising Systems with kac Potentials}, J. Statist. Phys, {\bf 82},
743--796 (1996).

\bibitem{BCP} Bertini L, Cancrini N. and Posta G. {\it On the dynamical behavior of the ABC model}, arXiv:1104.0822

%

\bibitem{CCELM0}  Carlen E. A., Carvalho M. C.,    Esposito R.,  Lebowitz J. L. and  Marra R., {\it
Droplet minimizers for the Gates-Lebowitz-Penrose free energy functional}, Nonlinearity {\bf 22},  2919--2952
(2009)


\bibitem{GP} Gates D. J. and Penrose O., {\it The van der Waals
limit for classical systems.
I. A variational principle.} Commun. Math. Phys. {\bf 15} 255--276 (1969).
 
  
\bibitem{LP} Lebowitz J.L. , Penrose  O.,
{\it  Rigorous treatment of the Van der Waals
Maxwell theory of the liquid vapor transition}
   J. Math. Phys., {\bf 7},  98,
     (1966)
     
     
%

\bibitem{P}  Presutti E.,  {Scaling Limits in Statistical Mechanics and Microstructures in Continuum Mechanics}, Springer (2009)
\end{thebibliography}
\end{document}